\documentclass[letterpaper, 10 pt, conference]{ieeeconf}

\title{\LARGE \bf
Leveraging Classification Metrics for Quantitative System-Level Analysis with Temporal Logic Specifications
}

\author{Apurva Badithela, Tichakorn Wongpiromsarn and Richard M. Murray
\thanks{A. Badithela is a graduate student in Control and Dynamical Systems, Computing and Mathematical Sciences,
        California Institute of Technology, Pasadena, CA 91106, USA
        {\tt\small apurva@caltech.edu}}%
\thanks{T. Wongpiromsarn is with the Department of Computer Science, Iowa State University,
        Ames, IA 50011, USA
        {\tt\small nok@iastate.edu}}%
\thanks{R.M. Murray is with the Control and Dynamical Systems, Computing and Mathematical Sciences,
        California Institute of Technology, Pasadena, CA 91106, USA
        {\tt\small murray@cds.caltech.edu}}%
}

\usepackage{times}

\usepackage[bookmarks=true]{hyperref}
\usepackage{algorithm,algpseudocode}
\usepackage{amsfonts}
\usepackage{mathtools}
\usepackage{tabularx,booktabs,caption}
\newcolumntype{Z}{>{\raggedright}X}
\usepackage{bbm}
\usepackage{tikz}
\usepackage{subcaption}
\usepackage{color}
\usepackage{enumerate}

\newtheorem{problem}{Problem}

\newtheorem{lemma}{Lemma}
\newtheorem{proposition}{Proposition}

\newtheorem{remark}{Remark}
\newtheorem{ex}{Example}
\newtheorem{defin}{Definition}

\newcommand{\mM}{\mathcal{M}}

\newcommand{\mC}{\mathcal{C}}
\newcommand{\mS}{\mathcal{S}}
\usepackage{dsfont}
\usepackage{adjustbox}

\IEEEoverridecommandlockouts                             
\makeatletter
\DeclareRobustCommand\bigop[1]{%
  \mathop{\vphantom{\sum}\mathpalette\bigop@{#1}}\slimits@
}
\newcommand{\bigop@}[2]{%
  \vcenter{%
    \sbox\z@{$#1\sum$}%
    \hbox{\resizebox{\ifx#1\displaystyle.9\fi\dimexpr\ht\z@+\dp\z@}{!}{$\m@th#2$}}%
  }%
}
\makeatother

\begin{document}

\maketitle
\thispagestyle{empty}
\pagestyle{empty}



\begin{abstract}
In many autonomy applications, performance of perception algorithms is important for effective planning and control. In this paper, we introduce a framework for computing the probability of satisfaction of formal system specifications given a confusion matrix, a statistical average performance measure for multi-class classification. We define the probability of satisfaction of a linear temporal logic formula given a specific initial state of the agent and true state of the environment. 
Then, we present an algorithm to construct a Markov chain that represents the system behavior under the composition of the perception and control components such that the probability of the temporal logic formula computed over the Markov chain is consistent with the probability that the temporal logic formula is satisfied by our system.
We illustrate this approach on a simple example of a car with pedestrian on the sidewalk environment, and compute the probability of satisfaction of safety requirements for varying parameters of the vehicle. We also illustrate how satisfaction probability changes with varied precision and recall derived from the confusion matrix. Based on our results, we identify several opportunities for future work in developing quantitative system-level analysis that incorporates perception models.
\end{abstract}
\section{Introduction}
\label{sec:introduction}
Autonomous systems usually consist of interconnected components, including perception and control, as shown in Figure~\ref{fig:arch}.
The perception component observes the  environment, classifies the objects and relevant features, and creates a representation of the world around the vehicle.
The control component uses this information to compute a trajectory for the vehicle to follow and the corresponding actuation commands to keep the vehicle on the trajectory.

These components are typically designed under different principles.
For example, the perception component often relies on object classification that is based on machine learning (ML) algorithms such as convolutional neural networks to discriminate objects
between different classes.
These ML-based algorithms are often evaluated based on the performance measures such as accuracy, precision, and recall \cite{geron2019hands,wang2019consistent}.

On the other hand, formal methods have been employed to construct a provably correct controller given a system model and temporal logic specifications \cite{kress2009temporal,kloetzer2008fully,lahijanian2009probabilistic,raman2014model,wongpiromsarn2012receding}. 
The correctness guarantee, as typically specified using a temporal logic formula,
heavily relies on the assumption that the input (i.e., the
perceived world reported by the perception component) is perfect.
For example, if the perception component only reports the most likely class of each object,
the control component assumes that the reported class is correct.
Unfortunately, this assumption may not hold in most real-world systems. 

\begin{figure}
\centering
\includegraphics[width=0.9\linewidth]{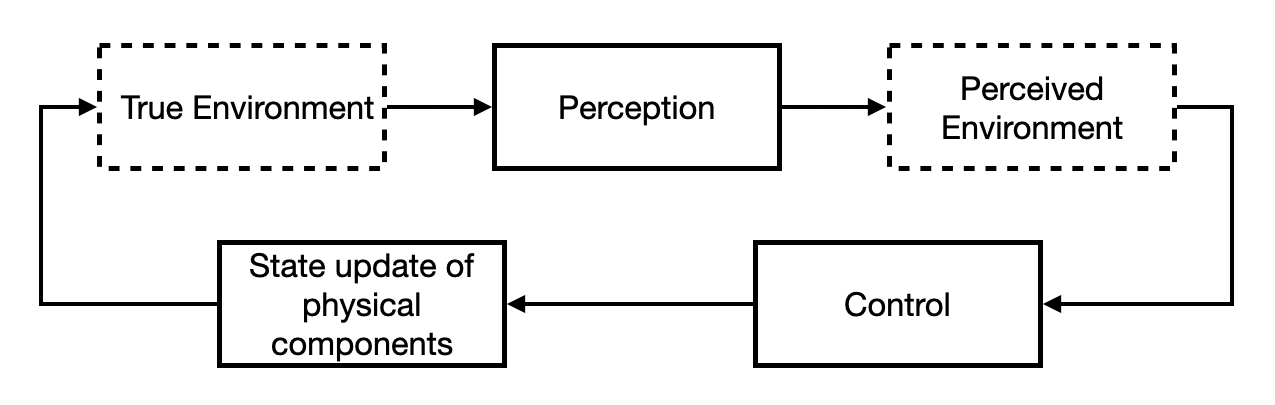}
\caption{Architecture of an autonomous system}
\label{fig:arch}
\end{figure}

In recent years, verifying neural networks with respect to safety and robustness properties has grown into an active research area~\cite{katz2017reluplex,fazlyab2019probabilistic,fazlyab2020safety,tran2020nnv}. Often, these methods apply to specific neural net architectures, such as those with piece-wise linear activation functions~\cite{katz2017reluplex}, or might require knowledge of the safe set in the output space of the neural network~\cite{fazlyab2019probabilistic,fazlyab2020safety}. Furthermore, these methods have been demonstrated on learning-based controllers with smaller input dimensions, and are not yet deployed for analysis of perception models. One reason for this is the difficulty in formally characterizing properties on ML-based perception models, as elaborated below. 

First, recent work demonstrates that it is not realistically feasible to formally specify properties reflecting human-level perception for perception models, in particular, classification ML models, due to the high dimensional nature of the input, such as pixels in an image~\cite{dreossi2018semantic}. Furthermore, Dreossi \textit{et al.} reason that not all misclassifications are the same; some are more likely to result in system-level failure, and therefore, it is necessary to adopt system-level specifications and contextual semantics in developing a framework for quantitative analysis and verification of perception models~\cite{dreossi2018semantic,seshia2018formal}. This has led to work on compositional analysis of perception models in finding system-level counter-examples~\cite{dreossi2019compositional}.

While there is work on evaluating performance of perception with temporal logic, those formal specifications are defined over image data streams, and must be manually formalized for each scenario / data stream~\cite{dokhanchi2018evaluating,balakrishnan2019specifying}. Often, there is high variability in the performance of perception models in seemingly similar environments, such as variations in sun angle~\cite{bauchwitz2020evaluating}. Thus, in a given scenario, manually constructed formal specifications might not exhaustively specify all of the desired properties for perception.

On the other hand, it is simpler to define system-level specifications, such as ``maintain a safe distance of \(5\) m from obstacles"~\cite{kress2008courteous,kress2008automatically,wongpiromsarn2011synthesis,dreossi2018semantic}. In this paper, we exploit the confusion matrix, a well-known performance measure in machine learning for classification tasks, to make system-level verification more practical. 
Moreover, various performance metrics such as accuracy, precision, or recall, can be derived from the confusion matrix, and perception algorithms are typically optimized to improve along these measures of performance~\cite{wang2019consistent}. Confusion matrices are used to characterize a model's performance in binary classification`\cite{yan2018binary}, multilabel classification~\cite{koyejo2015consistent}, and multiclass classification~\cite{narasimhan2015consistent}. Yet, the connection between these metrics and satisfaction of overall system specifications is not well understood. Furthermore, the precision-recall tradeoff is well-known in machine learning~\cite{geron2019hands}, that is, increasing precision typically results in reducing recall and vice-versa. To our knowledge, there is no systematic way of picking the right operating point of these competing objectives that accounts for system-level performance. 

The main contributions of this work are two-fold. First, we formally characterize the probability of satisfaction of a temporal logic formula over a trace given classification errors in perception. Secondly, we present an algorithm to construct a Markov chain representing the state evolution of the system, taking into account both the perception and control components. We prove that the satisfaction probability defined earlier is consistent with the probability of the temporal logic formula computed over the constructed Markov chain. We then employ existing probabilistic model checkers to compute the probability that the Markov chain satisfies the temporal logic formula. We also present empirical results on how different classification metrics affect the satisfaction probability of system-level specifications. 


    %

This paper is outlined as follows. In Section~\ref{sec:prelim}, we provide a description of the class of autonomous systems considered in our analysis, a running example of such a system and the corresponding specifications, and provide formal definitions for some performance measures of classification. In Section~\ref{sec:Problem_statement}, we present the problem statement and define the satisfaction probability. In Section~\ref{sec:Method}, we present an algorithm and a short proof describing how the confusion matrix can be used to determine satisfaction probability. In Section~\ref{sec:Example}, we present and discuss the main results of our analysis. Finally, we present our conclusions and directions for future work in Section~\ref{sec:conclusion}.

\section{Preliminaries}
\label{sec:prelim}

\subsection{System Description}
We consider a system comprising of an autonomous agent and its environment. Broadly, the autonomous agent is composed of two modules --- perception and control.
Typically, the main purpose of the perception module is object detection, which includes object localization and classification. (A more sophisticated perception component may include behavior prediction to better handle dynamic environments.)
In this paper, we focus our analysis on static environments, and leave dynamic environments for future work. Furthermore, the analysis focuses on the object classification inaccuracies. 
In several autonomy applications, the perception modules are neural network based. For the scope of this work, we treat the perception module as a black box that observes the true state \(x_e\) of an object in the environment and returns a classified label \(y_{e,t}\) of that object as the observed state of the environment at time step \(t\) to the control module.
\par
The control module receives the environment observation from the perception module to update the state \(s_t\) of the agent. Broadly, the control module is responsible for high-level (mission) planning, motion planning, and trajectory tracking. In this work, we assume that the controller takes the observed state \(y_{e,t}\) of the environment returned by the perception module as ground truth, and takes an action to update the state of the agent accordingly. Additionally, we focus on a discrete-state description of the system.
\par
The state of the system is the comprises of the state \(s_t\) of the autonomous agent and true state \(x_e\) of the environment. It is possible to formally specify high-level system requirements, and using this system description, we can ask the following question: \textit{Can we use performance metrics for classification-based perception to reason about the probability with which the high-level system specification will be satisfied?}

\subsection{Example} Here we introduce a specific example corresponding to the above system description that we will use throughout the paper. 
\begin{ex}
\label{ex:system}
Consider an autonomous car driving on a road as the agent and a sidewalk environment. We assume a static environment that comprises of the type of object on the sidewalk. The true state of the environment \(x_e\) can be one of three types --- a pedestrian denoted by ped, some other object that is not a pedestrian, denoted by obj, and the case where there is nothing on the sidewalk, denoted by empty. As illustrated in Figure~\ref{fig:Pedestrian}, the location of the sidewalk is fixed at step \(k\), for some \(k\) far from the starting position of the car such to allow for realizable correct-by-construction controllers for formal specifications that will be detailed in Example~\ref{ex:spec}. The car controller takes action to update its state at each discrete time step based on the observed state of the environment at that time step. The observed state of the environment, \(y_{e,t}\), is the class label that the perception module assigns to the object constituting the true environment at time step \(t\). We use the term system to mean the car and sidewalk environment combined.

\begin{figure}[h]
\centering
\includegraphics[scale=0.3]{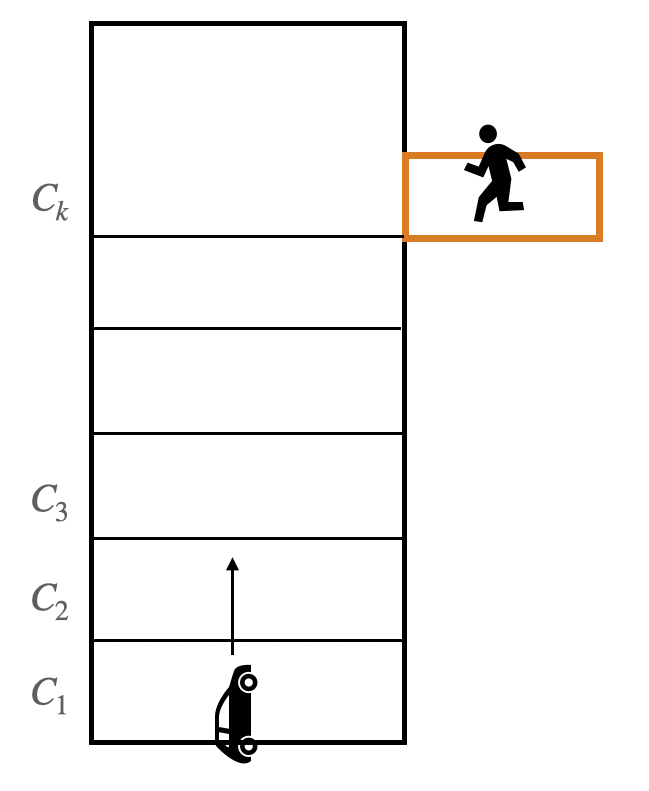}
\caption{Illustration of the car and pedestrian example}
\label{fig:Pedestrian}
\end{figure}
\end{ex}

\subsection{System Specification}
\label{subsec:intro_example}
We specify formal requirements on the system in Linear Temporal Logic (LTL). We introduce some preliminary notation before describing LTL; a more formal background on LTL can be found in~\cite{baier2008principles}. The agent is defined by variables \(V_A\), and the environment is defined by variables \(V_E\). The set of states of the agent is represented by the valuation of \(V_A\) and is denoted by \(S_A\), and the set of states of the environment is represented by the valuation of \(V_E\), denoted by \(S_E\). The set of states by the overall system is represented by \(S := S_A\times S_E\). Let \(AP\) denote the finite set of atomic propositions over the variables \(V_A\) and \(V_E\) of the agent and the environment. 

An LTL formula is defined by (a) a set of atomic propositions, (b) logical operators such as: negation (\(\neg\)), conjunction (\(\wedge\)), disjunction \((\vee)\), and implication (\(\implies\)), and (c) temporal operators such as: next (\(\bigcirc\)), eventually (\(\diamond\)), always (\(\square\)), and until (\(\mathcal{U}\)). The syntax of LTL is defined inductively as follows: (a) An atomic proposition \(p\) is an LTL formula, and (b) if \(\varphi\) and \(\psi\) are LTL formulae, then \(\neg \varphi\), \(\varphi \vee \psi\), \(\bigcirc \varphi\), \(\varphi \,\mathcal{U}\, \psi\) are also LTL formulae. LTL formulae with other temporal operators and combinations of logical connectives can be derived from these operators.

For an infinite trace \(\sigma = s_0s_1\ldots \in 2^{AP}\), an LTL formula \(\varphi\) defined over \(AP\), we use \(\sigma \models \varphi\) to mean that \(\sigma\) satisfies \(\varphi\). 
For example, the formula \(\varphi = \square p\) means that the atomic proposition \(p \in AP\) is satisfied at every state in the trace, i.e., \(\sigma \models \varphi\) if and only if \(p \in s_t, \forall t\).

\begin{ex}
\label{ex:spec}
We revisit Example~\ref{ex:system} to describe the discrete-state model of the car dynamics and its specifications. The state of the car is characterized by its position and speed, \(s_a := (x_c,v_c) \in S_A\). The car position is defined by the discrete cell it occupies, \(x_c = C_i\), where \(1 \leq i\leq N\) and \(N\) is the last cell index, and its forward speed, \(0 \leq v_c \leq V_{max}\). The perception module on the car can observe the sidewalk cell adjacent to \(C_k\) from any position on the road. The cell \(C_{k-1}\) is one road cell prior to the \(C_k\), which adjoins the sidewalk.
The overall system specifications include the states of the agent and the environment as,
\begin{enumerate}[(S1)]
     \item If the true state of the environment does not have a pedestrian, i.e, 
    \(x_e \neq ped\), then the car must not stop at \(C_{k-1}\).
     \item If \(x_e = ped\), the car must stop on \(C_{k-1}\).
     \item The agent should not stop at any cell \(C_i\), for all \(i \in \{1,\ldots, k-2\}\).\medskip
\end{enumerate}

Specifications (S\(1\)) and (S\(2\)) require the agent to stop at \(C_{k-1}\) only when there is a pedestrian on the sidewalk but not stop otherwise. We formally express the safety specifications on for the agent as follows.
\begin{equation}
\varphi_{1}= \square((x_e = ped)\vee \neg (x_c = C_{k-1} \wedge v_c = 0))\, , 
\label{eq:spec1}
\end{equation}
\begin{multline}
\varphi_{2}  = \square(\neg (x_e = ped)\vee \neg (x_c = C_{k-1}\vee \ldots \vee x_c = C_{N})\\ \vee (x_c = C_{k-1} \wedge v_c = 0))\, ,
\label{eq:spec2}
\end{multline}
\begin{equation}
    \varphi_3 = \square(\neg(\bigvee_{i = 1}^{k-1} (x_c = C_i \wedge v_c = 0)).
    \label{eq:spec3}
\end{equation} 
 The specifications (S\(1\)), (S\(2\)), and (S\(3\)) correspond to formulae $\varphi_1$, $\varphi_2$ and $\varphi_3$, respectively. \end{ex}

\subsection{Performance Measures for Classification in Perception}
For multi-class classification, we assume that statistical average performance is given in the form of a confusion matrix defined below. 
Let dataset \(\mathcal{D}\) be used to evaluate the perception module for \(n\)-class classification with classes \(c_1, \ldots, c_n\). Let \(\mathcal{D}_1, \ldots, \mathcal{D}_n\subset \mathcal{D}\) be a class-based partition of \(\mathcal{D}\) such that all datapoints \(x\in \mathcal{D}_i\) are labeled as \(c_i\) for all \(1\leq i\leq n\). For any datapoint \(x \in \mathcal{D}\), we write \(P(x) = c_i\) to denote that the predicted class of point \(x\) is \(c_i\), and we write \(T(x) = c_i\) to denote that the true class of \(x\) is \(c_i\). 

\begin{defin}[Confusion Matrix~\cite{geron2019hands}]
The \emph{confusion matrix} is a matrix \(C \in \mathbb{B}^{n\times n}\) such that for all \(i,j \in \{1,\ldots,n\}\),
\label{def:conf_matrix}
\end{defin}
\medskip
\begin{equation}
C(i,j) := \frac{\sum_{x \in \mathcal{D}_j} \mathbf{1}_{P(x) = c_i}}{|\mathcal{D}_j|},
\end{equation}
where \(\mathbf{1}_{P(x) = c_i}\) is the indicator function.
That is, for two classes, \(c_i\) and \(c_j\), \(C(i,j)\) represents the probability that a datapoint is classified as \(c_i\), given that its true class label is \(c_j\). 
For binary classification, with class labels \(c_1\) and \(c_2\), without loss of generality, \(C(1,1)\) represents the true positive rate, \(C(2,1)\) is the false negative rate, \(C(1,2)\) is the false positive rate, and \(C(2,2)\) is the true negative rate. 
Now, we give definitions for two simple measures of performance that can be derived from the confusion matrix. Without loss of generality, assume that the \(i^{th}\) class label is of interest.
\begin{defin}[Precision~\cite{geron2019hands}]
 Given the confusion matrix \(C\) for a multi-class classification, the \emph{precision} corresponding to class label \(c_i\) is
\label{def:precision}
\end{defin}
\begin{equation}
    P(i) = \frac{C(i,i)}{C(i,i) + \frac{\sum_{j\neq i}C(i,j)|\mathcal{D}_j|}{\sum_{j\neq i} |\mathcal{D}_j|}},
\label{eq:precision}
\end{equation}
where \(\frac{\sum_{j\neq i}C(i,j)|\mathcal{D}_j|}{\sum_{j\neq i} |\mathcal{D}_j|}\) is the false positive rate for class label \(c_i\), and \(C(i,i)\) is the true positive rate for class label \(c_i\). 
\begin{defin}[Recall~\cite{geron2019hands}]
Given the confusion matrix \(C\) for a multi-class classification, the \emph{recall} corresponding to class label \(c_i\) is as follows,
\label{def:recall}
\end{defin}
\begin{equation}
    R(i) = \frac{C(i,i)}{C(i,i) + \sum_{j\neq i}C(j,i)},
\label{eq:recall}
\end{equation}
where \(\sum_{j\neq i}C(j,i)\) is the false negative rate for class label \(c_i\).

Maximizing precision typically corresponds to minimizing false positives while maximizing recall corresponds to decreasing false negatives. However, there is an inherent trade-off in minimizing both false positives and false negatives for classification tasks~\cite{geron2019hands}, and often, a good operating point is found in an \textit{ad-hoc} manner. Typically, safety-critical systems are designed for optimizing recall, but as we will show, this is not always the best strategy to satisfy formal requirements.

\section{Problem Statement}
\label{sec:Problem_statement}
Here, we introduce and define the probability of satisfaction of an LTL formula starting from an initial state, given the true state of the environment. 
\begin{defin}[Transition Probability] Let \(s_1 = (s_{1,a},x_e)\), \(s_2 = (s_{2,a},x_e) \in S\) be two states of the system, \(x_e\) be the true class label of the environment, and \(C\) be the known confusion matrix associated with the agent's perception model. Let \(O(s_1, s_2)\) denote the set of environment observations \(y_e \in V_E\) that result in the agent controller transitioning from \(s_{1,a}\) to \(s_{2,a}\). The transition probability \(Pr: S \times S \rightarrow [0,1]\) is defined as,
\end{defin}
\begin{equation}
Pr(s_1, s_2) := \sum\limits_{y_e \in O(s_1, s_2)} C(y_e, x_e)\,.
    \label{eq:Pr}
\end{equation}
From the definition of the confusion matrix~\ref{def:conf_matrix}, it is trivial to check that \(\sum_{y_e \in O(s_1, s_2)} C(y_e, x_e) = 1\) and therefore, the transition probability has the range \([0,1]\). 
\begin{defin}[Paths] Choose a state \(s_0 = (s_{a,0}, x_e) \in S\) for a fixed true environment state \(x_e\). A finite path starting from \(s_0\) is a finite sequence of states \(\sigma (s_0) = s_0, s_1, \ldots,s_n\) for some \(n\geq 0\) such that the probability of transition between consecutive states, \(Pr(s_i, s_{i+1}) > 0\) for all \(0 \leq i < n\) such that \(s_i = (s_{a,i}, x_e) \in S\). Similarly, an infinite path \(\sigma = s_0, s_1, \ldots\) is an infinite sequence of states such that \(Pr(s_i, s_{i+1}) > 0\) for all \(i \geq 0\). We denote the set of all paths starting from \(s_0 \in S\) by \(Paths(s_0)\), and the set of all finite paths starting from \(s_0 \in S\) by \(Paths_{fin}(s_0)\). For an LTL formula \(\varphi\) on \(AP\), \(Paths_{\varphi}(s_0) \subset Paths(s_0)\) is the set of paths \(\sigma = s_0, s_1, \ldots\) such that \(\sigma_S \models \varphi\). 
\label{def:path}
\end{defin}
\subsection{Semantics for Probability of Satisfaction}
Now, we define probability of satisfaction of a temporal logic formula with respect to a formal specification based on the following definitions derived from~\cite{baier2008principles}.
Let \(\Omega = Paths(s_0)\) represents the set of all possible outcomes, that is, the set of all paths of the agent, starting from state \(s_0\). Let \(2^{\Omega}\) denote the powerset of \(\Omega\). Then, \((\Omega, 2^{\Omega})\) forms a \(\sigma\)-algebra. For a path \(\hat{\pi} = s_0, s_1, \ldots, s_n\)\(\in Paths_{fin}(s_0)\), we define a cylinder set as follows,
\begin{equation}
    Cyl(\hat{\pi}) = \{\pi \in Paths(s_0)|\hat{\pi}\in pref(\pi)\}. 
\end{equation}
Let \(\mathcal{C}_{s_0} = \{Cyl(\hat{\pi}) | \hat{\pi} \in Paths_{fin}(s_0)\}\). 
The following result can be found in ~\cite{baier2008principles}, and can be derived from the fundamental definition of a \(\sigma\)-algebra.
\begin{lemma}
The pair \((Paths(s_0), \,2^{\mathcal{C}_{s_0}})\) forms a \(\sigma\)-algebra, and is the smallest \(\sigma\)-algebra containing \(\mathcal{C}_{s_0}\).
\end{lemma}
The \(\sigma\)-algebra associated with \(s_0\) is (\(Paths(s_0), 2^{\mathcal{C}_{s_0}}\)). Then, there exists a unique probability measure \(\mathbb{P}_{s_0}\) such that 
\begin{equation}
    \mathbb{P}_{s_0}(Cyl(s_0, \ldots, s_n)) = \prod\limits_{0\leq i\leq n}Pr(s_i, s_{i+1}).
\end{equation}
\begin{defin} Consider an LTL formula \(\varphi\) over \(AP\) with the overall system starting at state \(s_0 = (s_{a,0}, x_e)\). Then, the probability that the system will satisfy the specification \(\varphi\) from the initial state \(s_0\) given the true state of the environment is,
\end{defin}
\begin{equation}
    \mathbb{P}(s_0 \models \varphi) := \sum\limits_{\sigma(s_0) \in \mathcal{S(\varphi)}}\mathbb{P}_{s_0}(Cyl(\sigma(s_0))),
\end{equation}
where \(\mathcal{S(\varphi)} := Paths_{fin}(s_0) \cap Paths_{\varphi}(s_0)\). Note that \(\mathcal{S(\varphi)}\) need not be a finite set, but has to be countable.

\subsection{Problem Formulation}
\begin{problem} Given a confusion matrix \(C\) for multi-class classification, a controller \(K\), a temporal logic formula \(\varphi\), the initial state of the agent \(s_{a,0}\), and the true state of the static environment \(x_{e}\), compute the probability \(\mathbb{P}(s_0 \models \varphi\)) that \(\varphi\) will be satisfied for a system trace \(\sigma\) starting from initial condition \(s_0 = (s_{a,0}, x_e)\)?
\label{prob}
\end{problem}

\begin{remark}If the probability distribution over the true state of the environment is known, we can compute the probability that the overall system satisfies the specification.
\end{remark}

\section{Confusion Matrices To Markov Chains}
\label{sec:Method}
Our approach to solving Problem~\ref{prob} is based on constructing a Markov chain that 
represents the state evolution of the agent, taking into account the interaction of the perception and the control components. This Markov chain is constructed for a particular true state of the environment. Given a Markov chain for the state evolution of the system, it is then straightforward to compute the probability of satisfying a temporal logic formula on the Markov chain from an arbitrary initial state~\cite{baier2008principles}. Probabilistic model checking can be used to compute the probability that the Markov chain satisfies the formula using existing tools such as PRISM~\cite{kwiatkowska2011prism} and Storm~\cite{dehnert2017storm}, which have been demonstrated to successfully analyze systems modeled by Markov chains with billions of states.

\begin{defin}[Markov Chain~\cite{baier2008principles}]
A discrete-time Markov chain is a tuple \(\mathcal{M} = (S, \mathbf{P}, \iota_{init}, AP, L)\), where \(S\) is a non-empty, countable set of states, \(\mathbb{P}:S\times S\rightarrow [0,1]\) is the \emph{transition probability function} such that for all states \(s \in S\), \(\Sigma_{s'\in S}\mathbf{P}(s,s') = 1\), \(\iota_{init}:S \rightarrow [0,1]\) is the initial distribution such that \(\Sigma_{s\in S}\iota_{init}(s) = 1\), \(AP\) is a set of atomic propositions, and \(L:S\rightarrow 2^{AP}\) is a labeling function. 
\end{defin}

\begin{algorithm}
\caption{Constructing Markov Chain}\label{alg1} \begin{algorithmic}[1] \Procedure{MC}{$S,K,O,C,x_e$}
\State \(\mM(s, s') = 0,\, \forall s,s'\in S\) \Comment{Initialize}
\For{$s_o \in S$}\Comment{Looping through states}
\State $\iota_{init}(s_0) = 1$ \Comment{Setting initial state probability}
\For{$y_e \in O$} \Comment{Possible observations of $x_e$}
      \State $s_f\leftarrow K(s_o, y_e)$ \Comment{Controller state update}
      \State $p\leftarrow C(y_e,x_e)$\Comment{From confusion matrix}
      \State $\mM (s_o, s_f) \leftarrow \mM (s_o, s_f) + p$ \Comment{Update}
   \EndFor\label{obs_loop}
\EndFor\label{state_loop}
   \State \textbf{return} $\mM$
\EndProcedure
\end{algorithmic}
\label{alg:construct_MC}
\end{algorithm}
\begin{remark}
While Algorithm~\ref{alg:construct_MC} constructs the overall system Markov chain for a deterministic controller, this construction can be easily extended to the case of a probabilistic controller by modifying lines 6--8 in Algorithm~\ref{alg:construct_MC} by including the probabilistic transitions of the controller.
\end{remark}
The \(\sigma\)-algebra of Markov chain \(\mathcal{M}\) is (\(Paths(\mathcal{M}, 2^{\mathcal{C}_\mathcal{M}})\)), where \(\mathcal{C}_{\mathcal{M}} = \{Cyl(\hat{\pi})|\hat{\pi}\in Paths_{fin}(\mM)\}\)~\cite{baier2008principles}. Let \(\mS_{\mM}(\varphi)\) denote all paths of the MC \(\mM\) in \(Paths_{fin}(\mM)\cap Paths(\mM)\).
\begin{defin}[Probability on a Markov Chain]
Given an LTL formula \(\varphi\) over \(AP\), a true state of the environment, \(x_e\), an initial system state, \(s_0 = (s_{a,0}, x_e)\), and a Markov chain \(\mathcal{M}\) describing the dynamics of the overall system, we denote the probability that the system will satisfy \(\varphi\) starting from state \(s_0\) as \(\mathbb{P}_{\mathcal{M}}(s_0 \models \varphi_s)\). This probability can be computed using standard techniques as described in~\cite{baier2008principles}.
\end{defin}
Note that the construction of \(\mM\) depends on the true environment state \(x_e\). 
\begin{proposition}
Given \(\varphi\) as a temporal logic formula over the agent and the environment states, true state of the environment \(x_e\), agent initial state \(s_{a,0}\), and a Markov chain \(\mathcal{M}\) constructed via Algorithm~\ref{alg1}, then \(\mathbb{P}(s_0\models \varphi)\) is equivalent to computing \(\mathbb{P}_{\mathcal{M}}(s_0 \models \varphi)\), where \(s_0 = (s_{a,0}, x_e)\).
\end{proposition}
\begin{proof}
We begin by considering the transition probabilities \(Pr\) and the transition probabilities on the Markov chain \(\mathbf{P}\). Since misclassification errors are the only source of non-determinism in the evolution of the agent state, by construction, we have that \(\mathbf{P}(s_i, s_j) = Pr(s_i, s_j)\) for some \(s_i, s_j \in S\).
Next, we compare the \(\sigma\)-algebra of Markov chain \(\mathcal{M}\) with the \(\sigma\)-algebra associated with state \(s_0\). By construction of the Markov chain, observe that any path \(p \in Paths(s_0)\) is also a path on the MC \(\mathcal{M}\), \(p \in Paths(\mM)\), and as a result \(\mathcal{C}_{s_0} \subset \mC_{\mM}\). Similarly, by construction, there is no finite trace on the Markov chain starting from \(s_0\), \(\sigma(s_0) \in \mathcal{S}_\mathcal{M}\) that is not in \(\mS(\varphi)\).
\par
\begin{align*}
    \mathbb{P}(s_0 \models \varphi) = & \sum\limits_{\sigma(s_0) \in \mS(\varphi)}\mathbb{P}_{s_0}(Cyl(\sigma(s_0))) \\
   =& \sum\limits_{\sigma(s_0) \in \mS(\varphi)} \prod\limits_{0\leq i < n} Pr(\sigma_i, \sigma_{i+1})\\
   =& \sum\limits_{\sigma(s_0) \in \mS(\varphi)} \prod\limits_{0\leq i < n} \mathbf{P}(\sigma_i, \sigma_{i+1}) \\
   =& \sum\limits_{\sigma(s_0) \in \mS_{\mM}(\varphi)} \prod\limits_{0\leq i < n} \mathbb{P}_{\mM}(Cyl(\sigma(s_0))) \\
   =& \mathbb{P}_{\mM}(s_0\models \varphi)
\end{align*}
\end{proof}

\section{Example}
\label{sec:Example}
\begin{figure*}[!htbp]
\centering
\begin{minipage}{.3\textwidth}
   \includegraphics[width=\linewidth,trim={0.9cm 0.0cm 0.8cm 0.0cm}]{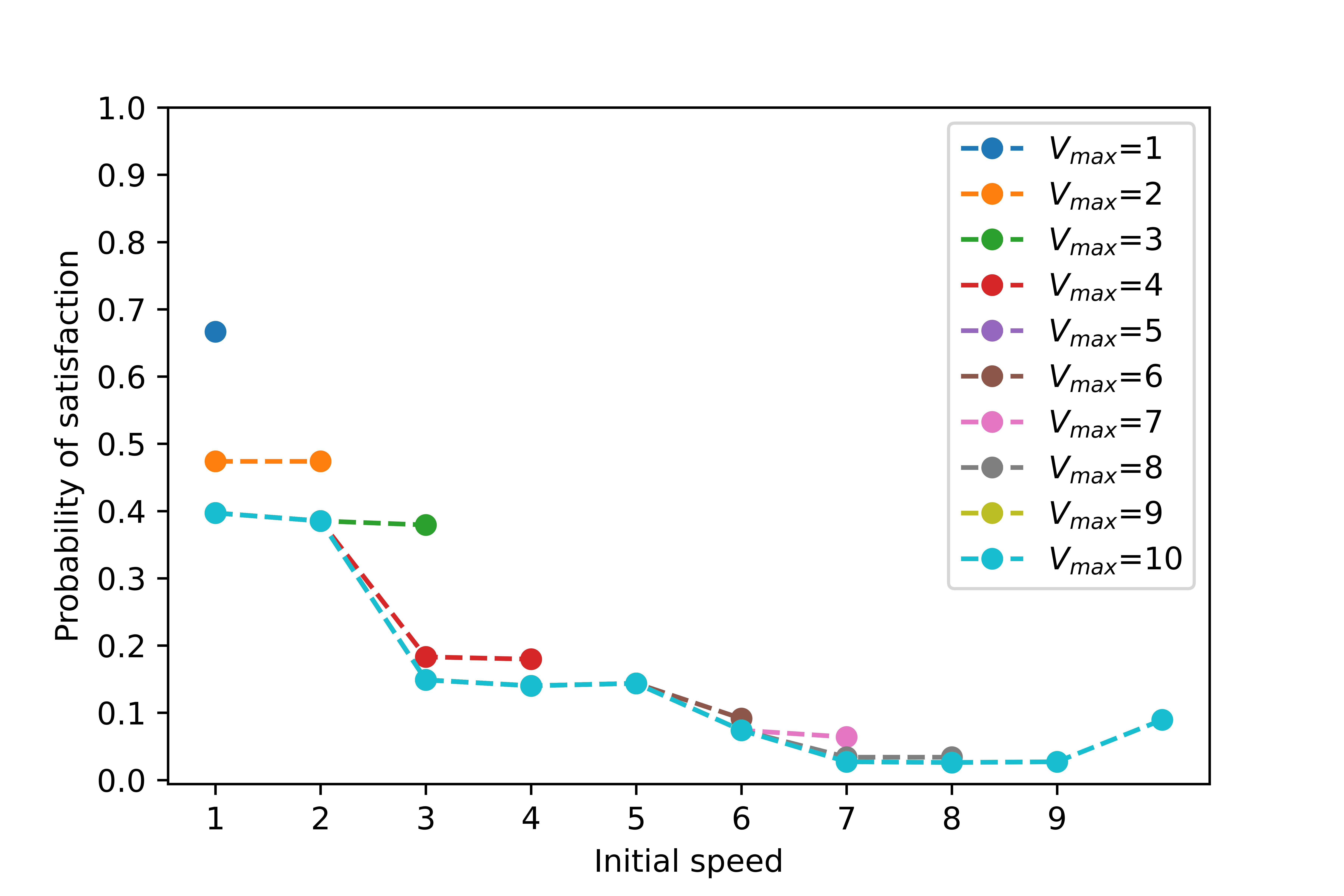}
   \subcaption{\label{fig:ped_plot} True environment: \emph{ped}}
  \end{minipage} 
  \begin{minipage}{.3\textwidth}
    \includegraphics[width=\linewidth,trim={0.9cm 0.0cm 0.8cm 0.0cm}]{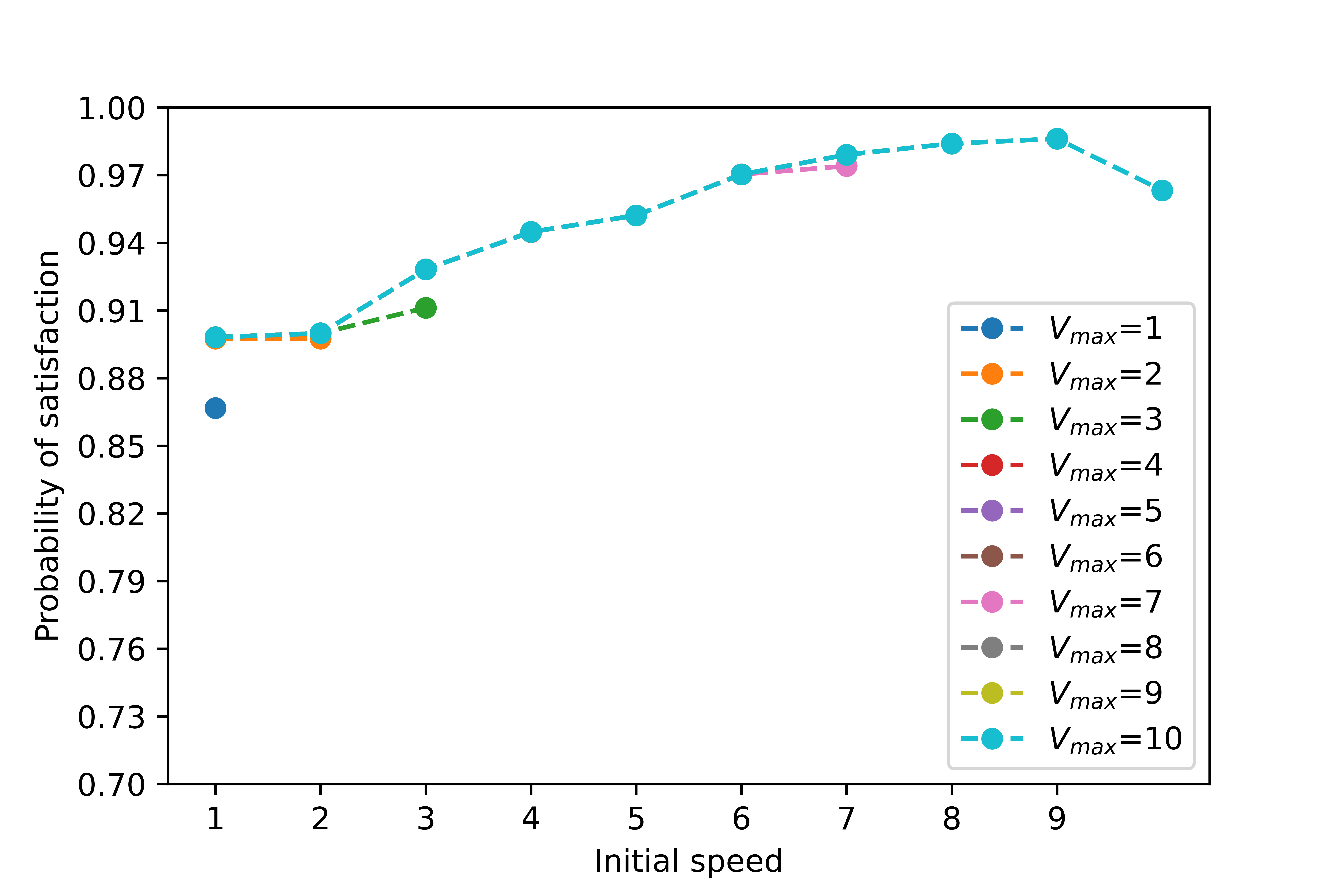}
    \subcaption{\label{fig:obj_plot} True environment: \emph{obj}}
  \end{minipage}
  \begin{minipage}{.3\textwidth}
    \includegraphics[width=\linewidth,trim={0.9cm 0.0cm 0.9cm 0.0cm}]{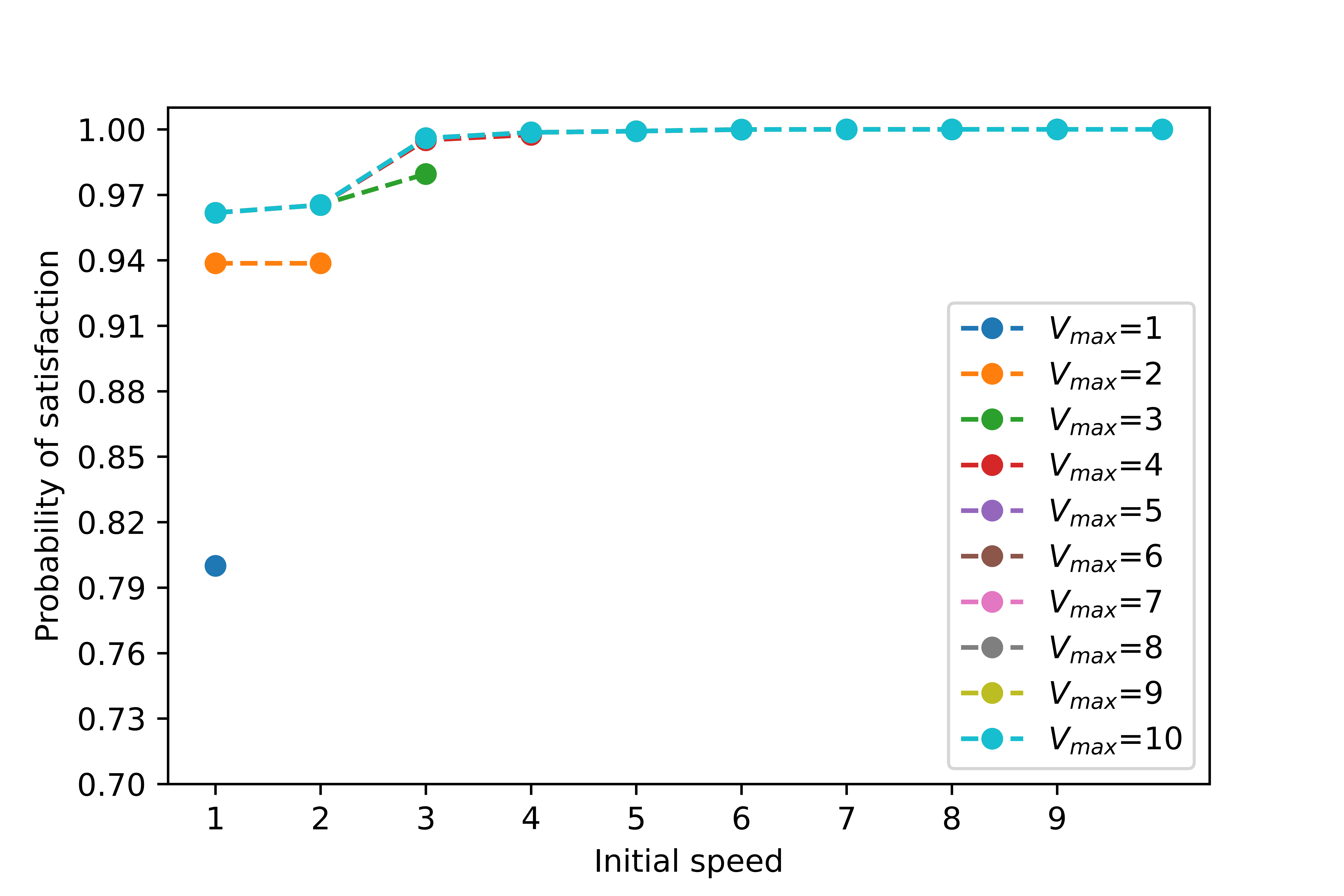}
    \subcaption{\label{fig:empty_plot} True environment: \emph{empty}}
  \end{minipage}
\caption{(a) Satisfaction probability that the car stops at \(C_{k-1}\) for \(x_e = ped\) under various initial speeds and maximum speeds \(V_{max}\) such that \(1\leq V_{max} \leq 10\). (b) Satisfaction probability that the car does not stop at \(C_{k-1}\) for \(x_e = obj\) under various initial speeds and maximum speeds \(V_{max}\) such that \(1\leq V_{max} \leq 10\). (c) Satisfaction probability that the car does not stop at \(C_{k-1}\) for \(x_e = empty\) under various initial speeds and maximum speeds \(V_{max}\) such that \(1\leq V_{max} \leq 10\).}
\label{fig:probability_plots}
\end{figure*}

\begin{figure*}[!htbp]
\centering
\begin{minipage}{.3\textwidth}
   \includegraphics[width=\linewidth,trim={0.9cm 0.0cm 0.8cm 0.0cm}]{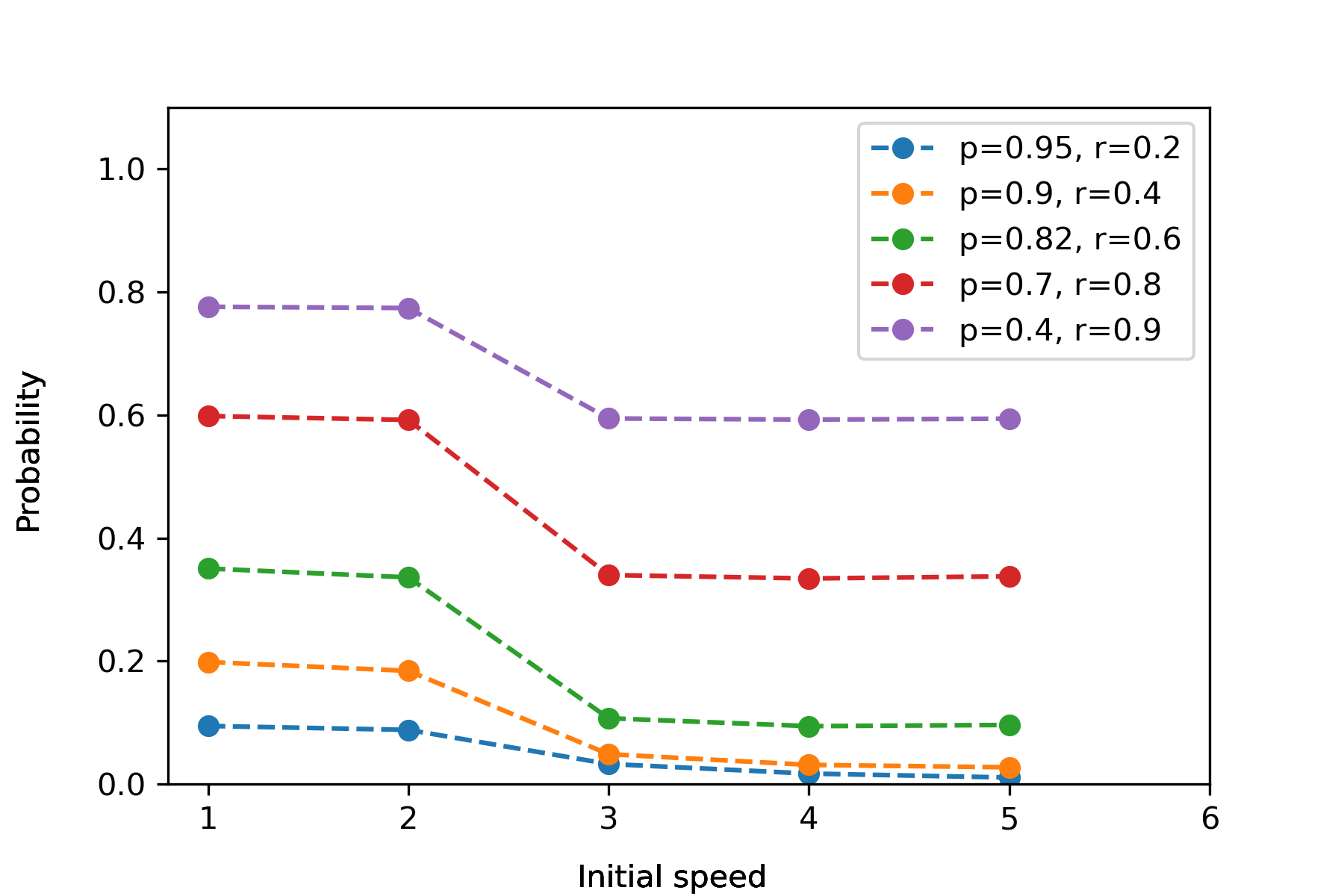}
   \subcaption{\label{fig:pr_ped} True environment: \emph{ped}}
  \end{minipage} 
  \begin{minipage}{.3\textwidth}
    \includegraphics[width=\linewidth,trim={0.9cm 0.0cm 0.8cm 0.0cm}]{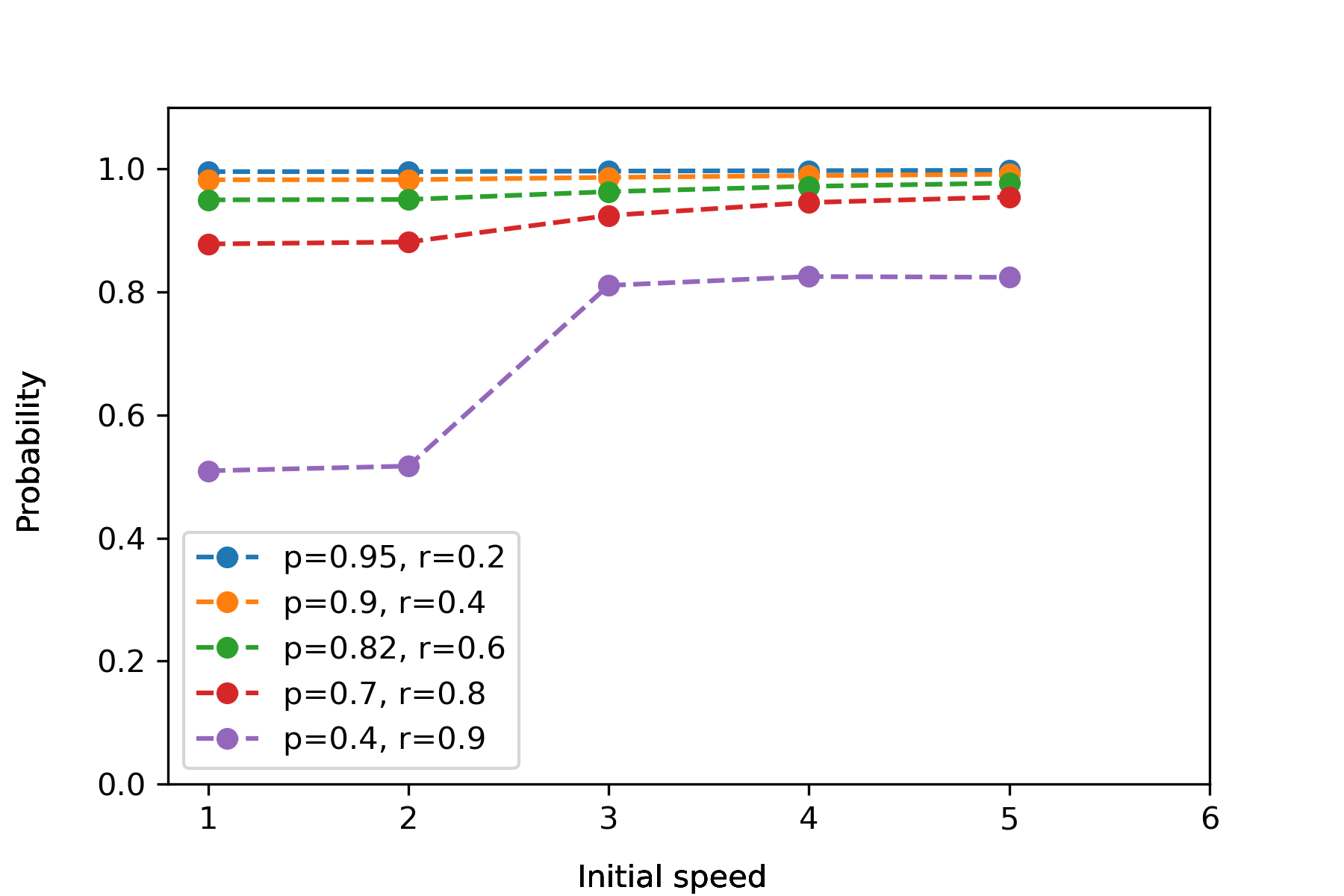}
    \subcaption{\label{fig:pr_obj} True environment: \emph{obj}}
  \end{minipage}
  \begin{minipage}{.3\textwidth}
    \includegraphics[width=\linewidth,trim={0.9cm 0.0cm 0.9cm 0.0cm}]{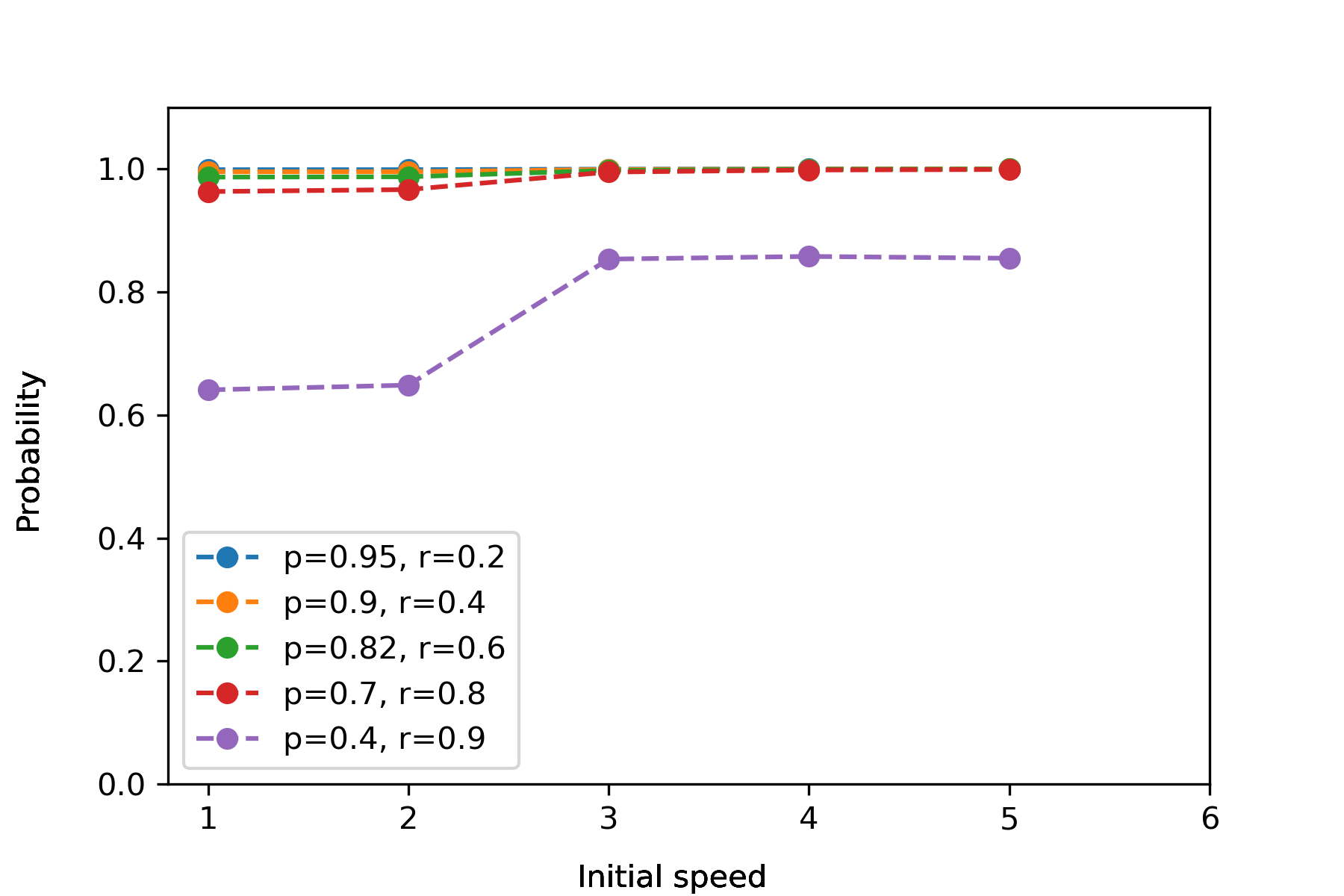}
    \subcaption{\label{fig:pr_empty} True environment: \emph{empty}}
  \end{minipage}
\caption{(a) Satisfaction probability that the car stops at \(C_{k-1}\) for \(x_e = ped\) and \(V_{max} = 5\) under various initial speeds and precision/recall pairs. (b) Satisfaction probability that the car does not stop at \(C_{k-1}\) for \(x_e = obj\) and \(V_{max} = 5\) under various initial speeds and precision/recall pairs. (c) Satisfaction probability that the car does not stop at \(C_{k-1}\) for \(x_e = empty\) and \(V_{max} = 5\) under various initial speeds and precision/recall pairs.}
\label{fig:prec_recall_figures}
\end{figure*}


\subsection{Autonomous Agent Controller}
We revisit car-sidewalk scenario described in Example~\ref{ex:system} with the corresponding  specifications listed in Example~\ref{ex:spec}, and describe the controller for the car. In this example, the observations of the environment are made at each time step as follows. Given that the true environment state is \(x_e = ped\), and at some time step \(t\), if the perception module incorrectly classifies the environment as \emph{obj}, it passes the observation \(y_{e,t} = obj\) to the control module. The car controller then takes an action at time step \(t\) based on the correct-by-construction controller corresponding to \(obj\).

If the true state of the environment has a static object \(x_e = obj\), then the agent slows down until it reaches a speed of \(v_c = 1\) and continues at that speed. If the true state of the environment is \emph{empty}, \(x_e = empty\), then the agent increases its speed until it reaches a speed of \(V_{max}\), and continues at that speed. If the true state of the environment is a pedestrian, \(x_e = ped\), we synthesize a correct-by-construction controller, using TuLiP~\cite{wongpiromsarn2011synthesis}, for the specifications \(\varphi_1\), \(\varphi_2\), and  \(\varphi_3\).
Note that our approach can be applied to any controller. We use this specific controller for illustration.
\par
Further, the controller is designed taking into account the following car dynamics. In the following equations, let \(k(i,v_c) := \max \{N, i+v_c\}\) for some \(1\leq i < N\).  
Then, the dynamics can be written as
\begin{align}
  \phantom{\square((x_c = C_i \wedge v_c = V_{max})}
  &\begin{aligned}
    \mathllap{\square((x_c = C_i \wedge v_c = 0)} & \rightarrow \bigcirc((x_c = C_i)\\
      &\qquad \wedge (v_c = 0 \vee v_c = 1))\, ,
    \label{eq:dyn1}
  \end{aligned}\\
  &\begin{aligned}
    \mathllap{\square((x_c = C_i \wedge  v_c = V_{max})} & \rightarrow \bigcirc(x_c = C_{k(i, V_{max})} \, \wedge\\
      & (v_c = V_{max} \vee v_c = V_{max}-1)))\, ,
    \label{eq:dyn2}
  \end{aligned}\\
  &\begin{aligned}
    \mathllap{\square((x_c = C_i \wedge v_c = v)} & \rightarrow \bigcirc(x_c = C_{k(i, v)} \wedge (v_c = v \, \vee \\
      & v_c = v-1 \vee v_c = v+1)))\,,
    \label{eq:dyn3}
   \end{aligned}
\end{align}
where and \(0 < v < V_{max}\) and \(C_N\) is the last road cell.

\label{subsec:control_synt}
\subsection{Results and Discussion}
\label{subsec:results}

We now present results of satisfaction probability of specifications formalized in equations~\eqref{eq:spec1}--~\eqref{eq:spec3} under varying initial conditions, confusion matrices, and true environment states. 
Once the Markov chain has been constructed, we used the model checker Storm to compute probabilities on the Markov chain~\cite{dehnert2017storm}. The implementations of Algorithm~\ref{alg1} and the results presented in this section are available online ~\cite{repo}. 
All the results in this section correspond to the setting of \(N = 65\), and the sidewalk located adjacent to cell \(C_{k}\), where \(k := 57\). The sidewalk position \(C_{57}\) was chosen so that the car has sufficient time to come to a complete stop when the initial speed is \(10\). 

\subsubsection{Initial Conditions and True Environment State}
For confusion matrix CM \(1\) as shown in Table~\ref{tab:cm_table}, the satisfaction probabilities are plotted for various true states of the environment -- \(ped\), \(obj\), and \(empty\) -- in Figures~\ref{fig:ped_plot}, ~\ref{fig:obj_plot}, and ~\ref{fig:empty_plot}, respectively. In Figure~\ref{fig:ped_plot}, the satisfaction probability indicates the probability that the car will stop at \(C_{k-1}\) since the true state of the environment contains a pedestrian. Figures~\ref{fig:obj_plot} and ~\ref{fig:empty_plot} show satisfaction probability for the car not stopping at \(C_{k-1}\) if the true state of the environment is an \(obj\) or \(empty\), respectively.

In Figures~\ref{fig:ped_plot}, the general trend is that for higher maximum speeds, the probability of satisfaction is lower; for \(V_{max} = 1\), the probability of satisfaction is the highest since the agent can bring itself to a stop in a single step according to its dynamics. Observe that for a fixed \(V_{max}\), the probability of satisfaction is monotonically decreasing in the initial speed of the agent, which is reasonable because higher initial speeds allow corresponds to a lower probability of recovering from perception errors due to dynamics of the car. The exception to this is that for \(V_{max} = 10\) and an initial speed of \(10\), the satisfaction probability is slightly higher than for an initial speed of \(9\). This small increase in probability is due to the location of the sidewalk (next to \(C_{57}\)) in relation to the initial speed of \(10\), which leads the car to a stop at the earliest in \(56\) steps according to its dynamics described in Eq~\ref{eq:dyn1}--\ref{eq:dyn3}. Thus, unless the car speed is \(v_c = 1\), any misclassifications of \(x_e\) as \(y_e = obs\) in the car's trace actually help the car in reducing its speed. Thus, the probability that the car will stop is the probability that it observes either \emph{ped} or \emph{obj} in the first \(9\) steps, and finally makes the correct observation of \(ped\) in the last step, which is equivalent to \((\frac{12}{15})^9\frac{10}{15} \approx 0.895\). 

In Figures~\ref{fig:ped_plot}, note that \(V_{max}=10\) forms a lower bound for satisfaction probability for varying initial speeds. For instance, for initial speed of \(1\), the satisfaction probability is the same for all \(V_{max} \ \geq 3\). While this might seem unusual, this observation can be explained as follows. If the initial speed is low enough compared to the maximum speed, we would need to observe \(y_e = empty\) repeatedly to increase the speed significantly due to the controller design in the car. The probability that this perception error occurs in multiple consecutive steps is small, and therefore, does not produce a noticeable difference in satisfaction probability when the initial speed is relatively small compared to a range of large maximum speeds.

In Figures~\ref{fig:obj_plot} and ~\ref{fig:empty_plot}, we observe that satisfaction probability monotonically increases with \(V_{max}\), with \(V_{max} = 10\) forming an upper bound on satisfaction probability. With a higher \(V_{max}\), the car could potentially increase its speed to higher speeds than with a lower \(V_{max}\), thus making it harder to stop. For instance in Figure~\ref{fig:obj_plot}, \(V_{max} = 1\) has the lowest satisfaction probability since the car can bring itself to a stop in one step; the probability that it will not stop at \(C_{k-1}\) is if it makes observations of \(y_e = obj\) or \(y_e = empty\), which from the confusion matrix CM \(1\) have the probability of \(\frac{13}{15} \approx 0.8667\). In Figure~\ref{fig:obj_plot}, the slight dip in probability from initial speed of \(9\) to \(10\) for \(V_{max} = 10\) can be explained using the same corresponding argument in Figure~\ref{fig:ped_plot} because controllers for \(ped\) and \(obj\) observations work to lower the speed of the car. In contrast, in Figure~\ref{fig:empty_plot}, there is no such dip in probability because the correct observation of \(y_e = empty\) works to keep the car speed higher to \(V_{max}\), and while misclassifications lead to lower the speeds, they have a small probability of occuring in several consecutive steps. Moreover, since the controller for \(empty\) works to increase the speed to \(V_{max}\), and the controller for \(obj\) works to decrease the speed until \(v_c = 1\), the probability that the car will not stop is higher in Figure~\ref{fig:empty_plot} than in Figure~\ref{fig:obj_plot}.



\subsubsection{Precision/Recall Tradeoff}
Often in autonomous driving applications, maximizing recall is prioritized over precision for safety purposes. In our example, maximizing recall would correspond with increasing tendency to stop at \(C_{k-1}\), even if \(x_e \neq ped\). In Figures~\ref{fig:prec_recall_figures}, we show how varying precision/recall affects the probability of satisfaction for \(V_{max} = 5\). These precision/recall pairs were chosen to reflect the general precision/recall tradeoff trends for classification tasks~\cite{geron2019hands}. For the results presented in this paper, we construct a confusion matrix as a function of precision (\(p\)) and recall (\(r\)) as shown in \(CM(p,r)\) of Table~\ref{tab:cm_table}. Note that these precision/recall pairs are in reference to the class label \emph{ped}.
\begin{table}[h!]
\caption{Confusion Matrices used in simulation}\label{tab:cm_table}
\setlength\tabcolsep{3.5pt} 
\begin{tabularx}{\columnwidth}{@{} Z *{6}{c} @{}}
\toprule 
Predicted & \multicolumn{3}{c}{True (CM 1)} & \multicolumn{3}{c@{}}{True (CM(p,r))}\\
\cmidrule(lr){2-4} \cmidrule(l){5-7} 
& \emph{ped}  & \emph{obj} & \emph{empty} & \emph{ped}  & \emph{obj} & \emph{empty}\\
\midrule 
\emph{ped}  &  10/15 & 2/15 & 3/15 & TP & FP/2 & FP/2   \\ 
\emph{obj} &  2/15 & 11/15 & 2/15 & FN/2 & 4TN/10 & TN/10  \\
\emph{empty} & 3/15 & 2/15 & 10/15 & FN/2 & TN/10 & 4TN/10  \\
\bottomrule 
\end{tabularx}
\end{table}
In Table~\ref{tab:cm_table}, \(TP\), \(FP\), \(TN\), \(FN\) are the number of true positives, false positives, true negatives, and false negatives, respectively, of the \(ped\) class label. These are derived from precision \(p\) and recall \(r\) as follows,
\begin{equation}
\begin{aligned}
    TP &= r\,, & FP &= TP(\frac{1}{p}-1)\,,\\
    TN &= 2-FP\,, & FN &= 1-TP\,.
\end{aligned}
\label{eq:rates}
\end{equation}
Note that this is one of many possible confusion matrices that could be constructed; we have chosen one of them for illustration, and we use it consistently across all precision/recall pairs. 

Figure~\ref{fig:pr_ped} shows the probability that the car will stop at \(C_{k-1}\) given that \(x_e = ped\) for varied precision/recall. We observe that for higher precision and lower recall, satisfaction probability is lower compared to that of lower precision and higher recall. For a fixed precision/recall setting, the satisfaction probability is monotonically decreasing with initial speed. This indicates that it is easier for the car to recover from misclassification errors for lower initial speeds, and can be reasonably explained by the car dynamics since the car can only increase or decrease its speed by 1 unit at every step. At lower initial speeds, the car needs to make several misclassification errors consecutively to continue increasing its speed, which has a low probability over the length of the trace. In contrast, at higher speeds, the car moves a greater distance in one time step, and therefore has fewer opportunities for recovering from misclassification errors.

Figures~\ref{fig:pr_obj} and ~\ref{fig:pr_empty} shows satisfaction probability when \(x_e = obj\) and \(x_e = empty\), respectively, under for different precision/recall pairs. This satisfaction probability indicates the probability that the car will not stop at \(C_{k-1}\). We observe that for lower recall and higher precision, the probability that the car will not stop is higher, and vice-versa for higher recall and lower precision. This corresponds to our intuition that for higher recall, and consequently lower precision, the car is more likely to stop even when there is no pedestrian. The satisfaction probability increases monotonically with initial speed. First, due to its dynamics, when the car of has a lower initial speed, it makes more observations at slower speeds before reaching \(C_{k-1}\) compared to when starting at higher initial speeds. Secondly, if a car starting at lower initial speed has to progressively to a higher speed \(v_c\), it needs to make several consecutive observations of \(y_e = empty\), which has a low probability as reasoned above. For these two reasons, with lower initial speeds, the car can make more misclassifications due to more observations, and the car tends to continue at a lower speed over its trace. Thus, because it takes fewer steps to bring the car to a stop if it is traveling with a low speed \(v_c\), it is easier to bring the car to a stop at \(C_{k-1}\).

Although both Figures~\ref{fig:pr_obj} and ~\ref{fig:pr_empty} represent the satisfaction probability that the car will not stop, Figure~\ref{fig:pr_empty} shows that the satisfaction probability is higher when the true environment is \emph{empty}  compared to Figure~\ref{fig:pr_obj}, where the true environment is \emph{obj}. This difference can be explained by the controllers for the \emph{obj} and \emph{empty} environments. If \(x_e = obj\), the car controller slows down the car by 1 unit unless \(v_c = 1\), and for \(x_e = empty\), the car controller speeds up the car by 1 unit unless \(v_c = V_{max}\). As reasoned above, at lower speeds, it is easier for the car to stop with a few misclassifications of the environment as \(ped\), thus leading to lower satisfaction probability in the case of \(x_e = obj\).


\section{CONCLUSION} 
\label{sec:conclusion}
In this work, we present preliminary results towards establishing synergy between performance metrics popular in classification and formal system requirements in autonomy. Specifically, we first define the probability of satisfaction of a temporal logic formula, and then present a simple algorithm that computes the satisfaction probability from a confusion matrix for classification tasks. In addition to observing that mis-classifications could lead to agent traces that violate the overall system requirements, we can also compute non-trivial probabilities of satisfaction of overall system requirements. Furthermore, we observe that due to the precision/recall tradeoff in classification algorithms, it is infeasible to satisfy all system requirements by maximizing either one of those perception performance measures. 
\par
This preliminary analysis opens several questions for future work, some of which include --- using formal requirements to characterize a class of confusion matrices that would be optimal for maximizing the satisfaction probability of those requirements, and optimizing for precision and recall in a manner that is consistent with formal system requirements. Further research is required to study performance metrics of perception tasks such as behavior prediction, such as in scenarios with dynamic and reactive environments, in the context of formal system requirements. 
\section{ACKNOWLEDGMENTS}
Apurva Badithela and Richard Murray acknowledge funding from AFOSR Test and Evaluation program, grant FA9550-19-1-0302. 

\bibliographystyle{IEEEtran}
\bibliography{references}

 \newcommand{\noop}[1]{}
\begin{thebibliography}{10}
\providecommand{\url}[1]{#1}
\csname url@samestyle\endcsname
\providecommand{\newblock}{\relax}
\providecommand{\bibinfo}[2]{#2}
\providecommand{\BIBentrySTDinterwordspacing}{\spaceskip=0pt\relax}
\providecommand{\BIBentryALTinterwordstretchfactor}{4}
\providecommand{\BIBentryALTinterwordspacing}{\spaceskip=\fontdimen2\font plus
\BIBentryALTinterwordstretchfactor\fontdimen3\font minus
  \fontdimen4\font\relax}
\providecommand{\BIBforeignlanguage}[2]{{%
\expandafter\ifx\csname l@#1\endcsname\relax
\typeout{** WARNING: IEEEtran.bst: No hyphenation pattern has been}%
\typeout{** loaded for the language `#1'. Using the pattern for}%
\typeout{** the default language instead.}%
\else
\language=\csname l@#1\endcsname
\fi
#2}}
\providecommand{\BIBdecl}{\relax}
\BIBdecl

\bibitem{geron2019hands}
A.~G{\'e}ron, \emph{Hands-on {m}achine {l}earning with Scikit-Learn, Keras, and
  TensorFlow: Concepts, Tools, and Techniques to Build Intelligent
  Systems}.\hskip 1em plus 0.5em minus 0.4em\relax O'Reilly Media, 2019.

\bibitem{wang2019consistent}
X.~Wang, R.~Li, B.~Yan, and O.~Koyejo, ``Consistent classification with
  generalized metrics,'' \emph{arXiv preprint arXiv:1908.09057}, 2019.

\bibitem{kress2009temporal}
H.~Kress-Gazit, G.~E. Fainekos, and G.~J. Pappas, ``Temporal-logic-based
  reactive mission and motion planning,'' \emph{IEEE Transactions on Robotics},
  vol.~25, no.~6, pp. 1370--1381, 2009.

\bibitem{kloetzer2008fully}
M.~Kloetzer and C.~Belta, ``A fully automated framework for control of linear
  systems from temporal logic specifications,'' \emph{IEEE Transactions on
  Automatic Control}, vol.~53, no.~1, pp. 287--297, 2008.

\bibitem{lahijanian2009probabilistic}
M.~Lahijanian, S.~B. Andersson, and C.~Belta, ``A probabilistic approach for
  control of a stochastic system from {LTL} specifications,'' in
  \emph{Proceedings of the 48h IEEE Conference on Decision and Control (CDC)
  held jointly with 2009 28th Chinese Control Conference}.\hskip 1em plus 0.5em
  minus 0.4em\relax IEEE, 2009, pp. 2236--2241.

\bibitem{raman2014model}
V.~Raman, A.~Donz{\'e}, M.~Maasoumy, R.~M. Murray, A.~Sangiovanni-Vincentelli,
  and S.~A. Seshia, ``Model predictive control with signal temporal logic
  specifications,'' in \emph{53rd IEEE Conference on Decision and
  Control}.\hskip 1em plus 0.5em minus 0.4em\relax IEEE, 2014, pp. 81--87.

\bibitem{wongpiromsarn2012receding}
T.~Wongpiromsarn, U.~Topcu, and R.~M. Murray, ``Receding horizon temporal logic
  planning,'' \emph{IEEE Transactions on Automatic Control}, vol.~57, no.~11,
  pp. 2817--2830, 2012.

\bibitem{katz2017reluplex}
G.~Katz, C.~Barrett, D.~L. Dill, K.~Julian, and M.~J. Kochenderfer, ``Reluplex:
  An efficient {SMT} solver for verifying deep neural networks,'' in
  \emph{International Conference on Computer Aided Verification}.\hskip 1em
  plus 0.5em minus 0.4em\relax Springer, 2017, pp. 97--117.

\bibitem{fazlyab2019probabilistic}
M.~Fazlyab, M.~Morari, and G.~J. Pappas, ``Probabilistic verification and
  reachability analysis of neural networks via semidefinite programming,'' in
  \emph{2019 IEEE 58th Conference on Decision and Control (CDC)}.\hskip 1em
  plus 0.5em minus 0.4em\relax IEEE, 2019, pp. 2726--2731.

\bibitem{fazlyab2020safety}
------, ``Safety verification and robustness analysis of neural networks via
  quadratic constraints and semidefinite programming,'' \emph{IEEE Transactions
  on Automatic Control}, 2020.

\bibitem{tran2020nnv}
H.-D. Tran, X.~Yang, D.~M. Lopez, P.~Musau, L.~V. Nguyen, W.~Xiang, S.~Bak, and
  T.~T. Johnson, ``{NNV}: The neural network verification tool for deep neural
  networks and learning-enabled cyber-physical systems,'' in
  \emph{International Conference on Computer Aided Verification}.\hskip 1em
  plus 0.5em minus 0.4em\relax Springer, 2020, pp. 3--17.

\bibitem{dreossi2018semantic}
T.~Dreossi, S.~Jha, and S.~A. Seshia, ``Semantic adversarial deep learning,''
  in \emph{International Conference on Computer Aided Verification}.\hskip 1em
  plus 0.5em minus 0.4em\relax Springer, 2018, pp. 3--26.

\bibitem{seshia2018formal}
S.~A. Seshia, A.~Desai, T.~Dreossi, D.~J. Fremont, S.~Ghosh, E.~Kim,
  S.~Shivakumar, M.~Vazquez-Chanlatte, and X.~Yue, ``Formal specification for
  deep neural networks,'' in \emph{International Symposium on Automated
  Technology for Verification and Analysis}.\hskip 1em plus 0.5em minus
  0.4em\relax Springer, 2018, pp. 20--34.

\bibitem{dreossi2019compositional}
T.~Dreossi, A.~Donz{\'e}, and S.~A. Seshia, ``Compositional falsification of
  cyber-physical systems with machine learning components,'' \emph{Journal of
  Automated Reasoning}, vol.~63, no.~4, pp. 1031--1053, 2019.

\bibitem{dokhanchi2018evaluating}
A.~Dokhanchi, H.~B. Amor, J.~V. Deshmukh, and G.~Fainekos, ``Evaluating
  perception systems for autonomous vehicles using quality temporal logic,'' in
  \emph{International Conference on Runtime Verification}.\hskip 1em plus 0.5em
  minus 0.4em\relax Springer, 2018, pp. 409--416.

\bibitem{balakrishnan2019specifying}
A.~Balakrishnan, A.~G. Puranic, X.~Qin, A.~Dokhanchi, J.~V. Deshmukh, H.~B.
  Amor, and G.~Fainekos, ``Specifying and evaluating quality metrics for
  vision-based perception systems,'' in \emph{2019 Design, Automation \& Test
  in Europe Conference \& Exhibition (DATE)}.\hskip 1em plus 0.5em minus
  0.4em\relax IEEE, 2019, pp. 1433--1438.

\bibitem{bauchwitz2020evaluating}
B.~Bauchwitz and M.~Cummings, ``Evaluating the reliability of {T}esla model 3
  driver assist functions,'' 2020.

\bibitem{kress2008courteous}
H.~Kress-Gazit, D.~C. Conner, H.~Choset, A.~A. Rizzi, and G.~J. Pappas,
  ``Courteous cars,'' \emph{IEEE Robotics \& Automation Magazine}, vol.~15,
  no.~1, pp. 30--38, 2008.

\bibitem{kress2008automatically}
H.~Kress-Gazit and G.~J. Pappas, ``Automatically synthesizing a planning and
  control subsystem for the {DARPA} {U}rban {C}hallenge,'' in \emph{2008 IEEE
  International Conference on Automation Science and Engineering}.\hskip 1em
  plus 0.5em minus 0.4em\relax IEEE, 2008, pp. 766--771.

\bibitem{wongpiromsarn2011synthesis}
T.~Wongpiromsarn, S.~Karaman, and E.~Frazzoli, ``Synthesis of provably correct
  controllers for autonomous vehicles in urban environments,'' in \emph{2011
  14th International IEEE Conference on Intelligent Transportation Systems
  (ITSC)}.\hskip 1em plus 0.5em minus 0.4em\relax IEEE, 2011, pp. 1168--1173.

\bibitem{yan2018binary}
B.~Yan, S.~Koyejo, K.~Zhong, and P.~Ravikumar, ``Binary classification with
  karmic, threshold-quasi-concave metrics,'' in \emph{International Conference
  on Machine Learning}.\hskip 1em plus 0.5em minus 0.4em\relax PMLR, 2018, pp.
  5531--5540.

\bibitem{koyejo2015consistent}
O.~Koyejo, N.~Natarajan, P.~Ravikumar, and I.~S. Dhillon, ``Consistent
  multilabel classification.'' in \emph{NIPS}, vol.~29, 2015, pp. 3321--3329.

\bibitem{narasimhan2015consistent}
H.~Narasimhan, H.~Ramaswamy, A.~Saha, and S.~Agarwal, ``Consistent multiclass
  algorithms for complex performance measures,'' in \emph{International
  Conference on Machine Learning}.\hskip 1em plus 0.5em minus 0.4em\relax PMLR,
  2015, pp. 2398--2407.

\bibitem{baier2008principles}
C.~Baier and J.-P. Katoen, \emph{Principles of {M}odel {C}hecking}.\hskip 1em
  plus 0.5em minus 0.4em\relax MIT press, 2008.

\bibitem{kwiatkowska2011prism}
M.~Kwiatkowska, G.~Norman, and D.~Parker, ``Prism 4.0: Verification of
  probabilistic real-time systems,'' in \emph{International conference on
  computer aided verification}.\hskip 1em plus 0.5em minus 0.4em\relax
  Springer, 2011, pp. 585--591.

\bibitem{dehnert2017storm}
C.~Dehnert, S.~Junges, J.-P. Katoen, and M.~Volk, ``A {S}torm is coming: A
  modern probabilistic model checker,'' in \emph{International Conference on
  Computer Aided Verification}.\hskip 1em plus 0.5em minus 0.4em\relax
  Springer, 2017, pp. 592--600.

\bibitem{repo}
\url{https://github.com/abadithela/validate_perception_metrics}.

\end{thebibliography}
\end{document}